\documentclass[reqno]{amsart}

\usepackage{amssymb}
\usepackage{amsthm}

\usepackage{mathtools}
\usepackage{ stmaryrd }
\usepackage{cite}
\usepackage{mathabx}
\usepackage{amssymb}
\usepackage{cite}
\usepackage{mathabx} 
\usepackage{hyperref} 
\usepackage[utf8]{inputenc}
\usepackage{enumerate}
\usepackage{amsmath}
\usepackage{amssymb}
\usepackage[T1]{fontenc} 
\usepackage[utf8]{inputenc} 

\usepackage{color,soul}
\definecolor{ItalianApricot}{rgb}{1,0.7,0.5}
\sethlcolor{ItalianApricot}

\theoremstyle{plain}
\newtheorem{thm}{Theorem}[section]

\newtheorem{lem}[thm]{Lemma}
\newtheorem{cor}[thm]{Corollary}

\theoremstyle{definition}
\newtheorem{defn}[thm]{Definition}

\theoremstyle{remark}
\newtheorem{remark}[thm]{Remark}

\numberwithin{equation}{section}

\renewcommand{\epsilon}{\varepsilon}
\renewcommand{\phi}{\varphi}

\begin{document}

\title[Prefix-free quantum Kolmogorov complexity]{  Prefix-free quantum Kolmogorov complexity}

\date{\today}

\author{Tejas Bhojraj}
 \address[Tejas Bhojraj]{Department of Mathematics, University of Wisconsin--Madison, 480 Lincoln Dr., Madison, WI 53706, USA}
 \email{bhojraj@math.wisc.edu}


\maketitle
\begin{abstract}
    We introduce quantum-K ($QK$), a measure of the descriptive complexity of density matrices using classical prefix-free Turing machines and show that the initial segments of weak Solovay random and quantum Schnorr random states are incompressible in the sense of $QK$. Many properties enjoyed by prefix-free Kolmogorov complexity ($K$) have analogous versions for $QK$; notably a counting condition.
    
    Several connections between Solovay randomness and $K$, including the Chaitin type characterization of Solovay randomness, carry over to those between weak Solovay randomness and $QK$. We work towards a Levin\textendash Schnorr type characterization of weak Solovay randomness in terms of $QK$.
    
    Schnorr randomness has a Levin\textendash Schnorr characterization using $K_C$; a version of $K$ using a computable measure machine, $C$. We similarly define $QK_C$, a version of $QK$. Quantum Schnorr randomness is shown to have a  Levin\textendash Schnorr and a Chaitin type characterization using $QK_C$. The latter implies a Chaitin type characterization of classical Schnorr randomness using $K_C$. 
\end{abstract}

\tableofcontents
\section{Introduction and Overview}
The theory of computation has been extended to the quantum setting; a notable example being the conception of a quantum Turing machine \cite{Mller2007QuantumKC,doi:10.1137/S0097539796300921}. Similarly, algorithmic information theory has been generalized to the quantum realm \cite{Vitnyi2001QuantumKC, Berthiaume:2001:QKC:2942985.2943376, Mller2007QuantumKC,2014JMP....55h2205B, brudno}. It hence seems natural to extend algorithmic randomness, the study of randomness of infinite bitstrings using notions from computation and information, to the quantum realm \cite{unpublished, qpl, bhojraj2020quantum}. 

A coherent, infinite sequence of qubits (called a \emph{state}\cite{unpublished}), is the quantum version of an infinite bitstring. With this notion of a state in hand, one can define randomness for states similarly to how randomness is defined for infinite bitstrings. 

Martin-L{\"o}f randomness (which is equivalent to Solovay randomness) and Schnorr randomness for infinite bitstrings, defined using the concept of `effective null sets', have characterizations in terms of initial segment prefix-free Kolmogorov complexity (denoted by $K$) \cite{misc,misc1, DBLP:series/txtcs/Calude02}; the initial segments of random infinite bitstrings are incompressible in the sense of $K$. 

Quantum Solovay randomness and quantum Schnorr randomness for states have been defined using the quantum versions of effective null sets \cite{unpublished, bhojraj2020quantum}. Analogously to the classical situation, one may explore the connections between quantum Solovay randomness and quantum Schnorr randomness and the initial segment descriptive complexity of states (a length $n$ initial segment of a state is a density matrix on $n$ qubits). 

Motivated by this, we asked whether there is a quantum analogue of $K$ which yields a characterization of quantum Solovay and quantum Schnorr randomness. We define $QK$, a complexity measure for density matrices based on prefix-free, classical Turing machines. The abbreviation $QK$ stands for `quantum-K', reflecting our intention of developing a quantum analogue of $K$, the classical prefix-free Kolmogorov complexity.

To the best of our knowledge, all notions of quantum Kolmogorov complexity developed so far, with one exception\cite{Vitnyi2001QuantumKC}, exclusively use machines which are not prefix-free (plain classical machines or quantum Turing machines) \cite{doi:10.1137/S0097539796300921, Berthiaume:2001:QKC:2942985.2943376,Mller2007QuantumKC}. $K_Q$, a notion developed in \cite{Vitnyi2001QuantumKC} uses a quantum Turing machine, $Q$ together with the classical prefix-free Kolmogorov complexity in its definition.

We give an overview of the main points in the paper. 

In Section \ref{sec:2} we introduce quantum-K ($QK$) for density matrices and some of its properties. Theorem \ref{thm:67} (generalized in Lemma \ref{lem:K}) shows that $QK$ agrees with $K$ on the classical qubitstrings. Theorem \ref{thm:8} is a tight upper bound for $QK$ similar to that for $K$. Theorem \ref{thm:count} is a counting condition similar to that for $QC$ \cite{Berthiaume:2001:QKC:2942985.2943376}, $C$ and $K$ \cite{misc, misc1}.

Section \ref{sec:4} reviews quantum algorithmic randomness: a recently developed \cite{bhojraj2020quantum, unpublished} theory of randomness for states (infinite qubitstrings) using quantum versions of the classical `effectively null set'.

Section \ref{sec:3}, the main focus of this paper, connects $QK$ with two quantum algorithmic randomness notions: weak Solovay randomness and quantum Schnorr randomness, defined in Section \ref{sec:4}. Two important characterizations show that the initial segments of Martin-L{\"o}f randoms (equivalently, of Solovay randoms) are asymptotically incompressible in the sense of $K$: the Chaitin characterization (See \cite{Chaitin:1987:ITR:24912.24913} and theorem 3.2.21 in \cite{misc}),
\[X  \text{ is Martin-L{\"o}f random} \iff \text{lim}_{n}  K(X\upharpoonright n)-n= \infty,\]
and the Levin\textendash Schnorr characterization (See theorem 3.2.9 in \cite{misc}),
\[X  \text{ is Martin-L{\"o}f random} \iff \exists c \forall n [ K(X\upharpoonright n)>n-c].\]
(Characterizations having the former form  will be called `Chaitin type' and those having the latter form will be called `Levin\textendash Schnorr type').
We investigate the extent to which these classical characterizations carry over to weak Solovay randoms and $QK$.

Theorem \ref{thm:7} is a Chaitin type of characterization of 
weak Solovay randomness ($\rho$ is weak Solovay random $\iff$ lim$_{n}QK^{\epsilon}(\rho_{n})-n = \infty$).
This shows that the Levin\textendash Schnorr condition ($\forall n[QK^{\epsilon}(\rho_{n})>^{+}n]$) is implied by weak Solovay randomness. 

Theorem \ref{thm:39} shows both Chaitin and Levin\textendash Schnorr type characterizations of weak Solovay randomness when restricting attention to a specific class of states. It is worth noting that Theorem \ref{thm:39} uses the proof of our main result (Theorem 2.11) in \cite{bhojraj2020quantum}. 

For general states, subsection \ref{subsect:weak} shows that the Levin\textendash Schnorr condition implies something slightly weaker than weak Solovay randomness. 

While $K$ plays well with Solovay randomness, $K_C$, a version of $K$ using a computable measure machine, $C$ (a prefix-free Turing machine whose domain has computable Lebesgue measure) gives a Levin\textendash Schnorr  characterization of Schnorr randomness (See theorem 7.1.15 in \cite{misc1}). Motivated by this, we introduce $QK_C$ a version of $QK$ using computable measure machines in subsection \ref{schn}. 

It turns out that $QK_C$ yields not just a Levin\textendash Schnorr type (Theorem \ref{thm:schnor}), but also a Chaitin type (Theorem \ref{thm:Chaitin}) characterization of quantum Schnorr randomness.

Theorem \ref{thm:Chaitin} together with Theorem \ref{thm:677} and lemma 3.9 in \cite{bhojraj2020quantum}, implies that Schnorr randoms have a Chaitin type characterization in terms of $K_C$ (Theorem \ref{class}). So, results in the quantum realm imply a new result in the classical setting.

In summary, we introduce $QK$ and show that the initial segments of weak Solovay random and quantum Schnorr random states are incompressible in the sense of $QK$.

\section{The Definition and Properties of QK}
\label{sec:2}
 We assume familiarity with the notions of density matrix (See for example, \cite{Nielsen:2011:QCQ:1972505}), prefix-free Kolmogorov complexity ($K$) and $\mathbb{U}$, the universal prefix-free (or self-delimiting) Turing machine (See \cite{misc,misc1, DBLP:series/txtcs/Calude02}).
 
 The output of $\mathbb{U}$ can be interpreted as unordered tuples of complex algebraic vectors (equivalently, finite subsets of natural numbers). 
 
The notation $\mathbb{U}(\sigma)\downarrow = F$ means that $ \mathbb{U}(\sigma)$ outputs the index of $F$ with respect to some fixed canonical indexing of finite subsets of the naturals. We will never use an ordering on the elements of $F$ in any of our arguments: $F$ will be used to define an orthogonal projection : $\sum_{v \in F} |v\big>\big<v|$ which clearly does not depend on an ordering on $F$. 

As explained in \cite{unpublished}, the quantum analogue of a bitstring of length $n$ is a density matrix on $\mathbb{C}^{2^n}$.
 
For a density matrix, $\tau$, let $|\tau|$ denote the $s$ such that $\tau$ is a transformation on $\mathbb{C}^{2^s}$. For any $n$, $\mathbb{C}^{2^n}_{alg}$ is the space of elements of $\mathbb{C}^{2^n}$ with complex algebraic entries.

Logarithms will always be base 2. The notation $\leq^+ , \geq^+, =^+$ will be used for `upto additive constant' relations.

For $\epsilon >0 $, $QK^{\epsilon}(\tau)$ is defined to be
\begin{defn}
\label{defn:4}
$QK^{\epsilon}(\tau) := $ inf  $\{|\sigma|+$log$|F|: \mathbb{U}(\sigma)\downarrow = F$, a orthonormal set in $\mathbb{C}^{2^{|\tau|}}_{alg}$ and $\sum_{v \in F} \big<v |\tau |v \big> > \epsilon \}$

\end{defn}
 The term $\sum_{v \in F} \big<v |\tau |v \big>$ is the squared length of the `projection of $\tau$ onto span($F$)' which also equals the probability of getting an outcome of `1' when measuring $\tau$ with the observable given by the Hermitian projection onto span$(F)$\cite{Nielsen:2011:QCQ:1972505}. Although it is useful to intuitively think of $\sum_{v \in F} \big<v |\tau |v \big>$ as the `projection of $\tau$ onto span($F$)', we use quotes as $\tau$ is a convex combination of possibly \emph{multiple} unit vectors, while the notion `projection onto a subspace' refers usually to a single vector. 
 
 Note that for a given $\tau$, $QK^{\epsilon}(\tau)$ is determined by the classical prefix-free complexities and dimensions of those subspaces, span$(F)$, such that the projection of $\tau$ onto span$(F)$ has squared length atleast $\epsilon$. I.e., $QK^{\epsilon}(\tau)$ depends only on the $K$-complexities and ranks of those projective measurements of $\tau$ such that the probability of getting an outcome of `$1$' is atleast $\epsilon$. Roughly speaking, $QK^{\epsilon}(\tau)$ depends on the dimensions and prefix-free complexities of subspaces which are $\epsilon$ `close' to $\tau$. 
 
 This is in contrast to $QC^{\epsilon}(\tau)$ which depends on the quantum complexities of density matrices, not classical prefix-free complexities of subspaces, which are $\epsilon$ close to $\tau$ (Recall that $QC$ is based on quantum Turing machines) \cite{Berthiaume:2001:QKC:2942985.2943376}. Also, while the rank of the approximating projection is taken into consideration in $QK$, the rank of the approximating density matrix is not taken into account in $QC$.
 
 So, $QC^{\epsilon}(\tau)$ quantifies the quantum complexity of approximating $\tau$ by density matrices upto $\epsilon$ while $QK^{\epsilon}(\tau)$ measures the sum of the prefix-free complexity and the logarithm of the dimension of subspaces $\epsilon$ close to $\tau$. 

A test demonstrating the quantum non-randomness of a state, $\rho$ uses computable sequences of projections of `small rank' which are $\epsilon$-close to initial segments (density matrices) of $\rho$ (See Section \ref{sec:4}). It hence seems plausible that a complexity measure for a density matrix, $\tau$ must reflect the complexities and ranks of projections $\epsilon$-close to $\tau$ in order to play well with quantum randomness notions for states.

We mention that our $QK$ is entirely different from the $ \overline{QK_M} $ and $\overline{QK^{\delta}_M}$ notions defined in Definition 3.1.1 in \cite{Mller2007QuantumKC} using quantum Turing machines.

$QK^{\epsilon}$ would not be a `natural' complexity notion for density matrices if the following theorem did not hold:
\begin{thm}
\label{thm:67}
 Fix a rational $\epsilon$. $K(\sigma)= QK^{\epsilon}(|\sigma\big>\big<\sigma|)$ holds for all classical bitstrings $\sigma$, upto an additive constant depending only on $\epsilon$.
\end{thm}
We isolate here a simple but useful property which will be used for proving Theorem \ref{thm:67}. 
\begin{lem}
\label{lem:30}
Let $n$ be a natural number, $E=(e_{i})_{i=1}^{2^{n}}$ be any orthonormal basis for $\mathbb{C}^{2^{n}}$ and $F$ be any Hermitian projection matrix acting on $\mathbb{C}^{2^{n}}$. For any $\delta>0$, let  \[S^{\delta}_{E,F}:= \{e_{i} \in E: \big<e_{i}|F|e_{i}\big> > \delta \}.\]
Then, $|S^{\delta}_{E,F}| < \delta^{-1} \text{Tr}(F)$.
\end{lem}
\begin{proof}
Note that since $F$ is a Hermitian projection, $\big<e_{i}|F|e_{i}\big>=\big<Fe_{i}|F e_{i}\big> = |F e_{i}|^{2}\geq 0$. So,
\[\delta|S^{\delta}_{E,F}| < \sum_{e_{i} \in S^{\delta}_{E,F}}\big<e_{i}|F|e_{i}\big> \leq \sum_{i \leq 2^{n}}\big<e_{i}|F|e_{i}\big> = \text{Tr}(F).\]
\end{proof}
\begin{proof}

We now prove Theorem \ref{thm:67}, the idea of which is as follows: Given a classical bitstring and a subspace `close' to it, we find a subspace spanned only by classical bitstrings `close' to this subspace. Then we compress each of the spanning classical strings and show that the string we began with must be one of these.
Fix a rational $\epsilon$. Consider the machine $P$ doing the following:
\begin{enumerate}
    \item On input $\pi$, $P$ searches for $\pi=\sigma \tau$ such that 
    $\mathbb{U}(\sigma)\downarrow = F$, an orthonormal set, $F \subseteq \mathbb{C}^{2^{n}}_{alg}$ for some $n$ and $|\tau|= \lceil \text{log}(\epsilon^{-1}|F|)\rceil$.
    \item Letting $O:= \sum_{v\in F} |v\big>\big<v|$ and $E$, the standard basis of $\mathbb{C}^{2^{n}}$, find the set $S^{\epsilon}_{E,O}$ from Lemma \ref{lem:30}.
    \item Take a canonical surjective map $g$ from the set of bitstrings of length $\lceil \text{log}(\epsilon^{-1}|F|)\rceil$ onto $S^{\epsilon}_{E,O}$. ($g$ exists since  $|S^{\epsilon}_{E,O}|<\epsilon^{-1}|F|$ by \ref{lem:30}).  Output $g(\tau)$.
\end{enumerate}
We first show that $P$ is prefix-free. Suppose $\pi$ and $\pi'$ are in the domain of $P$ and $\pi \preceq \pi'$. Then, $\pi=\sigma \tau$ and $\pi'=\sigma'\tau'$ and $\sigma$ and $\sigma'$ are in the domain of $\mathbb{U}$. $\pi \preceq \pi'$ implies that $\sigma \preceq \sigma'$ or $\sigma' \preceq \sigma$. But as $\mathbb{U}$ is prefix-free, $\sigma = \sigma'$ must hold. Since the computations $\mathbb{U}(\pi)$, and $\mathbb{U}(\pi')$ and not stuck forever at (1), it must be that $|\tau|= \lceil \text{log}(\epsilon^{-1}|F|)\rceil$ and $|\tau'|= \lceil \text{log}(\epsilon^{-1}|F'|)\rceil$ where $\mathbb{U}(\sigma)\downarrow=F=F'=\mathbb{U}(\sigma')\downarrow$. So, $\tau$ and $\tau'$ have the same length implying that $\pi=\pi'$.
\\
Now, let $\sigma \in 2^n$ be any classical bitstring. Let $\lambda$ and $F \subseteq \mathbb{C}^{2^{n}}_{alg}$ a orthonormal set such that $|\lambda|+ \text{log}(|F|)=QK^{\epsilon}(|\sigma\big>\big<\sigma|)$,  $\sum_{v \in F} \big<v |\sigma\big>\big<\sigma|v \big> > \epsilon$ and $\mathbb{U}(\lambda)=F$. Let $O:= \sum_{v\in F} |v\big>\big<v|$. Note that since
$\epsilon<\sum_{v \in F} \big<v |\sigma\big>\big<\sigma|v \big> = \big<\sigma|O|\sigma\big>$, $\sigma \in S^{\epsilon}_{E,O}$ where $E$ is the standard basis. Let $\tau$ be a length $ \lceil \text{log}(\epsilon^{-1}|F|)\rceil$ string such that $g(\tau)=\sigma$. Then, we see that $P(\lambda \tau) = \sigma $. \[\text{K}(\sigma)\leq^+   |\lambda|+|\tau|\leq^+  |\lambda|+\text{log}(\epsilon^{-1})+ \text{log}(|F|) \leq^+ QK^{\epsilon}(|\sigma\big>\big<\sigma|)\]
This establishes one direction. Note that the additive constant depends on $\epsilon$. The constant (zero) in the other direction turns out to be independent of $\epsilon$: Given some classical bitstring $\sigma$, let $\mathbb{U}(\pi)=\sigma$ and $|\pi|=\text{K}(\sigma)$. Then, letting $F=\{\sigma\}$ in \ref{defn:4}, $QK^{\epsilon}(\sigma)\leq QK^{1}(\sigma)\leq |\pi|=K(\sigma)$, for any $\epsilon>0$. 

\end{proof}
 \begin{defn}
\label{defn:78}
A `system' $B= ((b^{n}_{0},b^{n}_{1}))_{n\in \mathbb{N}}$ is a sequence of orthonormal bases for $\mathbb{C}^{2}$ such that each $b^{n}_{i}$ is complex algebraic and the sequence $ ((b^{n}_{0},b^{n}_{1}))_{n\in \mathbb{N}}$ is computable.
\end{defn} 
\begin{remark}\label{rem:K}
Let $B= ((b^{n}_{0},b^{n}_{1}))_{n\in \mathbb{N}}$ be a  system, as in \ref{defn:78}. Let $A_B$ be the set of all pure states, $\sigma$ such that $\sigma$ is a product tensor of elements from $B$. For example, $b^{1}_{0}\otimes b^{2}_{1}\otimes b^{3}_{1}\otimes b^{4}_{0} \in A_B$. Then, the previous theorem generalizes to the following: Fix a rational $\epsilon$, a $B$ and a $A_B$ as above. $K(\sigma)=^+ QK^{\epsilon}(|\sigma\big>\big<\sigma|)$ holds for all $\sigma \in A_B$, upto an additive constant depending only on $B$ and $\epsilon$. Here, $K(\sigma)$ is defined in the obvious way. For example, $K(b^{1}_{0}\otimes b^{2}_{1}\otimes b^{3}_{1}\otimes b^{4}_{0})=K(0110)$. This is proved by replacing $S^{\epsilon}_{E,O}$ with $S^{\epsilon}_{B,O}$ in the proof of Theorem \ref{thm:67}.
\end{remark}
The following lemma can be proved similarly to Theorem \ref{thm:67}.
\begin{lem}
\label{lem:K}
Fix a rational $\epsilon$ and let $(B_{n})_n$ be a computable sequence such that $B_n$ is a orthonormal basis for $\mathbb{C}^{2^{n}}$ composed of algebraic complex vectors. Then, for all $\sigma \in \bigcup_{n} B_{n}$, $K(\sigma)= QK^{\epsilon}(|\sigma\big>\big<\sigma|)$, upto an additive constant depending only on $\epsilon$ and $(B_{n})_n$.
\end{lem}
Note that $K(\sigma)$ is well-defined as $\sigma$ is complex algebraic.
The following Theorem \ref{thm:8} agrees nicely with the upper bound for $K$ in the classical setting: for all strings $x$, $K(x) \leq |x|+K(|x|)+1 $ (See theorem 2.2.9 in \cite{misc}.).
\begin{thm}
\label{thm:8}
There is a constant $d>0$ such that for any $\epsilon$ and any $\tau$, QK$^{\epsilon}(\tau) \leq |\tau| + $ K$(|\tau|) + d$.
\end{thm}
\begin{proof}
Let $k>1$. Let $P$ be the prefix-free Turing machine which on input $\pi$, such that $\mathbb{U}(\pi)=n$ outputs $E=(e_{i})_{i=1}^{2^{n}}$, the standard computational basis of $\mathbb{C}^{2^{n}}$. 
\end{proof}
It may seem that this upper bound, given by the apparently inefficient device of using $2^{|\tau|}$ many orthonormal vectors to approximate $\tau$, can be improved. However, the bound is tight by Theorem \ref{thm:67} together with the classical counting theorem (see \cite{misc1}, theorem 3.7.6.).
\\
As we shall see later, the unique tracial state $\tau = (\tau_{n})_{n\in \mathbb{N}}$ where for all $n$, $\tau_n$ is the $2^n$ by $2^n$ diagonal matrix with $2^{-n}$ along the diagonal is quantum Martin-L{\"o}f random. Theorem \ref{thm:31} shows that its initial segments achieve the upper bound given by Theorem \ref{thm:8}.
\begin{thm}
\label{thm:31}
Let $k$ be any natural number. There is a constant $t$ such that for all $n$, $QK^{2^{-k}}(\tau_{n})\geq n+K(n)-t$.
\end{thm}
\begin{proof} 
Fix a $k$ and suppose towards a contradiction that for all $t \in \mathbb{N}$, there is a $n_{t} $ such that $QK^{2^{-k}}(\tau_{n_{t}})< n_{t}+K(n_{t})-t$. So, for all $t$, there are $F_{t} \subseteq \mathbb{C}^{2^{n_{t}}}$ and $\sigma_{t}$ such that $\mathbb{U}(\sigma_{t})=F_{t}$ and

\[2^{-k}<\sum_{v\in F_{t}} \big<v|\tau_{n_{t}}|v\big>= 2^{-n_{t}}|F_{t}|,\]
and
\[|\sigma_{t}|+\text{log}(|F_{t}|) < n_{t}+K(n_{t})-t.\]
\\
Taking log on both sides of the first inequality and inserting in the second gives that for all $t$, $n_t$ and $\sigma_{t}$,
\begin{align}
    \label{eq:43}
    && t-k+|\sigma_{t}| <  K(n_{t}). 
\end{align}
Now, define a prefix-free machine $M$ as follows. On input $\pi$, $M$ checks if $\mathbb{U}(\pi)$ halts and outputs a orthonormal set $F\subseteq \mathbb{C}^{2^{n}}$ for some $n$. If so, then $M(\pi)=n$. Let r be the coding constant of $M$. Note that for all $t$, $M(\sigma_{t})=n_t$. So, $K(n_t) \leq |\sigma_{t}|+r.$
Together with \eqref{eq:43}, we have that for all $t$,
$t-k+|\sigma_{t}|< |\sigma_{t}|+r$. So, $t-k<r$ for all $t$, a contradiction.

\end{proof}

In contrast to Lemma \ref{lem:K}, we have,
\begin{lem}
Fix an $\epsilon>0$ and an $n\in \mathbb{N}$. It is not true that for all $\sigma$, complex algebraic pure states in $\mathbb{C}^{2^{n}}$, $QK^{\epsilon}(|\sigma\big>\big<\sigma|)=^{+} K(\sigma)$.
\end{lem}
\begin{proof}
Clearly, for all $\sigma$, complex algebraic pure states, $QK^{\epsilon}(|\sigma\big>\big<\sigma|)\leq^{+} K(\sigma)$ holds.
Suppose that for some $\epsilon$ and $n\in \mathbb{N}$, for all $\sigma \in \mathbb{C}^{2^{n}}_{alg}$, pure states, $QK^{\epsilon}(|\sigma\big>\big<\sigma|)\geq^{+} K(\sigma)$ holds.
By Theorem \ref{thm:8}, for all $\sigma \in \mathbb{C}^{2^{n}}_{alg}$, pure,  
$K(n)+n\geq^{+}QK^{\epsilon}(|\sigma\big>\big<\sigma|)\geq^{+} K(\sigma)$. This is a contradiction as there are only finitely many programs of length atmost $n+K(n)$ but there are infinitely many complex algebraic pure states, $\sigma$ of length $n$.
\end{proof}Analogously to $QC$ \cite{Berthiaume:2001:QKC:2942985.2943376,brudno} a `counting condition' also holds for $QK$: the cardinality of a orthonormal set of vectors with bounded complexity has an upper bound depending on the complexity bound. The counting condition for $QK$ is established in a different fashion than that for $QC$ (which uses entropy inequalities like  Holevo's-chi \cite{Berthiaume:2001:QKC:2942985.2943376} and Fanne's inequality \cite{brudno}). This reflects once again that $QC$ invloves approximating a density matrix by another density matrix while $QK$ involves `projecting' a density matrix onto a  subspace.
\begin{thm}
\label{thm:count}
Let $V=(v_{i})_{i=1}^{N} \subset \mathbb{C}^{2^{s}}$ be a collection of orthonormal vectors with $QK^{\epsilon}(|v_{i}\big>\big<v_{i}|)\leq B$ for all $i$. Then, $N\leq \epsilon^{-1}2^{B}$. 
\end{thm}
\begin{proof}
For each $v_i$, we have $\sigma_i$ and $F_i$,\[F_{i}= \sum_{t\in A_{i}} |t\big>\big<t|,\]
with $A_{i}\subset \mathbb{C}^{2^{s}}$ orthonormal, such that $\big<v_{i}|F_{i}|v_{i}\big>>\epsilon$, $\mathbb{U}(\sigma_{i})=F_i$ and $|\sigma_{i}|+$log$|A_{i}|\leq B$. Let $D\subseteq \{1,2\cdots N\}$ be maximal such that $F_{i}\neq F_j$ for $i,j$ in $D$. ($D \neq \{1,2\cdots N\}$ may hold as there may be $i,j$ with $F_{i}=F_j$). Let $F$ be the orthogonal projector onto the subspace spanned by $A:=\bigcup_{i \in D}A_{i}$. Then, $A$ has dimension atmost $\sum_{i \in D}|A_{i}|$. By $|A_{i}| \leq 2^{-|\sigma_{i}|}2^{B}$ for all $i$ and noting that $\sigma_{i}\neq \sigma_j$ for $i,j$ in $D$,
\[\text{Tr}(F) \leq \sum_{i \in D}|A_{i}| \leq \sum_{i \in D} 2^{-|\sigma_{i}|}2^{B}\leq 2^{B}\sum_{\sigma \in \text{dom}(\mathbb{U})}2^{-|\sigma|} \leq  2^B.\]

The reason behind summing over $i \in D$, rather than over $i \leq N$ was to get the second to last inequality. By the maximality of $D$, $A =\bigcup_{i \leq N}A_{i}$ and so, $A_i$ is a subspace of $A$ for all $i\leq N$. So, $\big<v_{i}|F|v_{i}\big>\geq \big<v_{i}|F_{i}|v_{i}\big> >\epsilon $ for all $i \leq N$. By orthonormality of $V$,

\[\epsilon N < \sum_{i} \big<v_{i}|F|v_{i}\big> \leq \text{Tr}(F).\]
\end{proof}

\section{Quantum algorithmic randomness}
\label{sec:4}
We briefly review quantum algorithmic randomness. All definitions in this section are from \cite{unpublished} or \cite{bhojraj2020quantum} unless indicated otherwise.

While it is clear what one means by an infinite sequence of bits, it is not immediately obvious how one would formalize the notion of an infinite sequence of qubits. To describe this, many authors have independently come up with the notion of a \emph{state} \cite{unpublished,article,brudno}. We will need the one given by Nies and Scholz \cite{unpublished}. 
\begin{defn}
A \emph{state}, $\rho=(\rho_n)_{n\in \mathbb{N}}$ is an 
infinite sequence of density matrices such that $\rho_{n} \in \mathbb{C}^{2^{n} \times 2^{n}}$ and $\forall n$,  $PT_{\mathbb{C}^{2}}(\rho_n)=\rho_{n-1}$.
\end{defn}
The idea is that $\rho$ represents an infinite sequence of qubits whose first $n$ qubits are given by $\rho_n$. Here, $PT_{\mathbb{C}^{2}}$ denotes the partial trace which `traces out' the last qubit from $\mathbb{C}^{2^n}$. The definition requires $\rho$ to be coherent in the sense that for all $n$, $\rho_n$, when `restricted' via the partial trace to its first $n-1$ qubits, has the same measurement statistics as the state on $n-1$ qubits given by $\rho_{n-1}$.
The following state will be the quantum analogue of Lebesgue measure.
\begin{defn}
\label{def:tr}
Let $\tau=(\tau_n)_{n\in \mathbb{N}}$ be the state given by setting $\tau_n = \otimes_{i=1}^n I$ where $I$ is the two by two identity matrix.
\end{defn}
\begin{defn}
A \emph{special projection} is a hermitian projection matrix with complex algebraic entries.
\end{defn}
Since the complex algebraic numbers (roots of polynomials with rational coefficients) have a computable presentation, we may identify a special projection with a natural number and hence talk about computable sequences of special projections.
\begin{defn}
\label{defn:sigclass}
A quantum $\Sigma_{1}^0$
set (or q-$\Sigma_{1}^0$
set for short) G is a computable
sequence of special projections $G=(p_{i})_{i\in \mathbb{N}}$ such that $p_i$ is $2^i$ by $2^i$ and range$(p_i \otimes I) \subseteq$ range $(p_{i+1})$ for all $i\in \mathbb{N}$. 

\end{defn}
\begin{defn}
If $\rho$ is a state and $G=(p_{n})_{n\in \mathbb{N}}$ a q-$\Sigma_{1}^0$
set as above, then $\rho(G):=\lim_{n}$ Tr$(\rho_{n}p_n)$.
\end{defn}

\begin{defn}
A \emph{quantum Martin-L{\"o}f test} (q-MLT) is a computable sequence, $(S_{m})_{m \in \mathbb{N}}$ of q-$\Sigma_{1}^0$ classes such that $\tau(S_m)$ is less than or equal to $2^{-m}$ for all m, where $\tau$ is as in Definition \ref{def:tr}.
\end{defn}

\begin{defn}
$\rho$ is \emph{q-MLR} if for any q-MLT $(S_{m})_{m \in \mathbb{N}}$, $\inf_{m \in \mathbb{N}}\rho(S_m)=0$.
\end{defn}
Roughly speaking, a state is q-MLR if it cannot be `detected by projective measurements of arbitrarily small rank'.
\begin{defn}
$\rho$ is said to \emph{fail  the q-MLT $(S_{m})_{m \in \mathbb{N}}$, at order $\delta$}, if $\inf_{m \in \mathbb{N}}\rho(S_m)>\delta$. $\rho$ is said to \emph{pass  the q-MLT $(S_{m})_{m \in \mathbb{N}}$ at order $\delta$} if it does not fail it at $\delta$. 
\end{defn}
So, $\rho$ is q-MLR if it passes all q-MLTs at all $\delta>0$. Quantum Martin-L{\"o}f randomness is modelled on the classical notion: 
An infinite bitstring $X$ is said to \emph{pass} the Martin-L{\"o}f test $(U_{n})_n$ if $X \notin \bigcap_{n}U_n$ and is said to be \emph{Martin-L{\"o}f random (MLR)} if it passes all Martin-L{\"o}f tests (See 3.2.1 in \cite{misc}). 

A related notion is Solovay randomness. A computable sequence of $\Sigma^{0}_1$ classes, $(S_{n})_n$ is a \emph{Solovay test} if $\sum_{n} \mu(S_{n})$, the sum of the Lebesgue measures of the $S_n$s is finite. An infinite bitstring $X$ \emph{passes} $(S_{n})_n$ if $X\in S_n$ for infinitely many $n$ (See 3.2.18 in \cite{misc}).

We obtain a notion of a quantum Solovay test by replacing `$\Sigma^{0}_1$ class' and `Lebesgue measure' in the definition of classical Solovay tests with `quantum-$\Sigma^{0}_1$ set' and $\tau$ respectively. The following definitions are from \cite{bhojraj2020quantum} unless indicated otherwise:
\begin{defn}\label{defn:1} 
A uniformly computable sequence of quantum-$\Sigma^{0}_1$ sets, $(S^{k})_{k\in\omega}$ is a \emph{ quantum-Solovay test} if  $\sum_{k\in \omega} \tau(S^{k}) <\infty.$\end{defn}\begin{defn}\label{defn:2}
For $0<\delta<1$, a state $\rho$ \emph{fails the Solovay test $(S^k)_{k\in\omega}$ at level $\delta$} if there are infinitely many $k$ such that $\rho(S^k)>\delta$.
\end{defn}
\begin{defn}
A state $\rho$ \emph{passes the Solovay test $(S^k)_{k\in\omega}$} if for all $\delta>0$, $\rho$ does not fail $(S^k)_{k\in\omega}$ at level $\delta$. I.e., lim$_{k}\rho(S^{k})=0$.
\end{defn}
\begin{defn}
A state $\rho$ is \emph{quantum Solovay random} if it passes all quantum Solovay tests.\end{defn}
It is remarkable that $X$ is MLR if and only if it passes all Solovay tests (See 3.2.19 in \cite{misc}). This is also true in the quantum realm\cite{bhojraj2020quantum}.

An \emph{interval Solovay test} is a Solovay test, $(S_{n})_n$ such that each $S_n$ is generated by a finite collection of strings (See 3.2.22 in \cite{misc}). Its quantum version is:

\begin{defn}\cite{unpublished}
A \emph{strong Solovay test} is a computable sequence of special projections $(S^{m})_{m}$ such that $\sum_{m}\tau(S^{m}) < \infty$. A state $\rho$ fails $(S^{m})_{m}$ at $\epsilon$ if for infinitely many $m$, $\rho(S^{m})>\epsilon$.
\end{defn}
\begin{defn}\cite{unpublished}
A state $\rho$ is \emph{weak Solovay random} if it passes all strong quantum Solovay tests.\end{defn}
It is open whether weak Solovay randomness is equivalent to q-MLR. We need the notion of a computable real number to talk about Schnorr randomness: For the purposes of this paper, a function, $f$ from the natural numbers to the rationals is said to be \emph{computable} if there is a Turing machine, $\phi$ such that on input $n$, $\phi$ halts and outputs $f(n)$ (See Theorem 5.1.2 in \cite{misc1}). Note here that we interpret the output of a Turing machine as a rational number.
\begin{defn}
A sequence $(a_n)_{n\in \mathbb{N}}$ is said to be \emph{computable} if there is a computable function $f$, such that $f(n)=a_n$.

\end{defn}

\begin{defn}
A real number $r$ is said to be \emph{computable} if there is a computable function $f$ such that for all $n$, $|f(n)-r|<2^{-n}$.
\end{defn}

By 7.2.21 and 7.2.22 in \cite{misc1}, a Schnorr test may be defined as:
\begin{defn}
A \emph{Schnorr test} is an interval Solovay test, $(S^{m})_{m}$ such that $\sum_{m}\mu(S^{m}) $ is a computable real number. 
\end{defn}
An infinite bitstring passes a Schnorr test if it does not fail it (using the same notion of failing as in the Solovay test). We mimic this notion in the quantum setting.
\begin{defn}
\label{defn:Schnorr}
A \emph{quantum Schnorr test} is a strong Solovay test, $(S^{m})_{m}$ such that $\sum_{m}\tau(S^{m}) $ is a computable real number. A state is quantum Schnorr random if it passes all Schnorr tests.
\end{defn}

\section{Relating QK to randomness }
\label{sec:3} 
\subsection{A Chaitin type result}

Theorem \ref{thm:7} is a Chaitin type characterization of the weak Solovay random states in terms of $QK$. (Chaitin's result in the classical setting says that an infinite bitstring $X$ is Solovay random if and only lim$_{n}K(X\upharpoonright n) - n = \infty$). 
\begin{thm}
\label{thm:7}
A state $\rho = (\rho_{n})_{n}$ is weak Solovay random if and only if \[\forall \epsilon>0 \forall c >0 \forall^{\infty} n  QK^{\epsilon}(\rho_{n}) \geq n+c.\]
\end{thm}

\begin{proof}
($\Longleftarrow)$: Suppose for a contradiction that $\rho$ fails a strong Solovay test $(S_m)_{m}$ at $\epsilon>0$. The idea will be to use the subspaces given by the $ S_m $s, to approximate $\rho$. More, precisely, the $F$ appearing in the definition of $QK^{\epsilon}$ will be the orthonormal vectors given by the projection $S_m$ for an appropriate $m$. The details are as follows.  Let $M$ be the prefix-free machine doing the following. On input $\sigma$, if $\mathbb{U}(\sigma)=m$, then output $(v_i)_{i}$ where \[S_{m}= \sum_{i} |v_{i}\big>\big<v_{i}|.\] Let $c_{M}$ be it's coding constant. Take an $m$ such that Tr$(\rho_{n_{m}}S_{m})>\epsilon$ (Notation: $n_{m}$ is the natural number $n$ such that $S_m$ is a projection on $n$ qubits.). By the choice of $m$, \[\text{QK}^{\epsilon}(\rho_{n_{m}}) \leq \text{K}(m)+ c_{M}+ \text{log}(2^{n_{m}}\tau(S_{m}))= n_{m}+ \text{K}(m) - f(m) + c_{M},\]
where $f$ is the function: $f(m)= -\text{log}(\tau(S_m))$. As $f$ is computable and as $\sum_{m} 2^{-f(m)} < \infty $ by the definition of a strong Solovay test, Lemma 3.12.2 in \cite{misc1} implies that for all $m$, $\text{K}(m) - f(m) \leq q$ for some constant $q$. Noting that we may assume the sequence $n_{m}$ to be strictly increasing in $m$ and  letting $c:= q+c_{M}+1$, we see that $\exists^{\infty} n$  such that $ \text{QK}^{\epsilon}(\rho_{n})< n+c.$ 
\\
($\Longrightarrow$):
Suppose toward a contradiction that there is a $\epsilon>0$ and a constant $c>0$ such that there are infinitely many $n$ with $ \text{QK}^{\epsilon}(\rho_{n})< n+c.$  Define a strong Solovay test $S$ as follows. Let $T$ be the set of all $\sigma$ such that $\mathbb{U}(\sigma)$ halts and outputs an orthonormal set $F_{\sigma}\subseteq  \mathbb{C}^{2^{n_{\sigma}}}$  such that $|\sigma|+$ log $|F_{\sigma}|< n_{\sigma} + c$. For all $\sigma \in T$, let \[P_{\sigma}:= \sum_{v\in F_{\sigma}} |v \big>\big<v |\] and let $S:= (P_{\sigma})_{\sigma \in T}$. For all $\sigma \in T$, $2^{|\sigma|}|F_{\sigma}|< 2^{n_{\sigma} + c}$. So, $ \tau(P_{\sigma})= 2^{-n_{\sigma}}|F_{\sigma}|< 2^{c-|\sigma| }$.  So,
\[\sum_{\sigma \in T}\tau(P_{\sigma}) < 2^{c}\sum_{\sigma \in T}2^{-|\sigma|} < 2^{c}\sum_{\sigma: \mathbb{U}(\sigma)\downarrow} 2^{-|\sigma|} < \infty, \]
since $\mathbb{U}$ is prefix-free. This shows that $S$ is a strong Solovay test. For any $n$ such that $ \text{QK}^{\epsilon}(\rho_{n})< n+c$, there is a $\sigma \in T$ such that Tr$(P_{\sigma} \rho_{n})>\epsilon$. So, $\rho$ fails $S$ at $\epsilon$.
\end{proof}
The following corollary shows the equivalence of weak Solovay and q-ML randomness for a specific type of states.
Let $B= ((b^{n}_{0},b^{n}_{1}))_{n\in \mathbb{N}}$ be a system ( Definition \ref{defn:78}). Let $A^{\infty}_B$ be the set of all states which are limits of elements from $A_B$ as in
\ref{rem:K}. For example, $b^{1}_{0}\otimes b^{2}_{1}\otimes b^{3}_{0}\otimes b^{4}_{1}\cdots =: \rho \in A^{\infty}_B$.

\begin{cor}
\label{cor}
For any $B$, weak Solovay randomness is equivalent to q-MLR on $A^{\infty}_B$.
\end{cor}
\begin{proof}
Fix a system $B= ((b^{n}_{0},b^{n}_{1}))_{n\in \mathbb{N}}$ and let $\rho\in A^{\infty}_B$ be weak Solovay random. Let $\rho'$ be the bitstring induced by $\rho$. I.e., for example if $\rho=b^{1}_{0}\otimes b^{2}_{1}\otimes b^{3}_{0}\otimes b^{4}_{1} \cdots$, then $\rho':= 0101\cdots$.  By Theorem \ref{thm:7}, for $\epsilon=0.5$ \[  \forall c >0 \forall^{\infty} n  QK^{0.5}(\rho_{n}) \geq n+c.\] By Remark \ref{rem:K}, $K(\rho'\upharpoonright n)= MK^{0.5}(\rho_{n})$ upto a constant depending only on $B$. So, 
\[  \forall c >0 \forall^{\infty} n  K(\rho'\upharpoonright n) \geq n+c.\]
By Chaitin's result, \cite{Chaitin:1987:ITR:24912.24913} $\rho'$ is MLR. Now, by an easy modification of 3.13 from \cite{unpublished}, $\rho$ is q-MLR. We already know that q-MLR implies weak Solovay randomness for any state from before.
\end{proof}

\subsection{Chaitin and Levin\textendash Schnorr type results}
It turns out that weak Solovay randomness is equivalent to q-MLR and has both Chaitin ((3) in Theorem \ref{thm:39} ) and Levin\textendash Schnorr ((4) in Theorem \ref{thm:39}) type characterizations in terms of $QK$ when the states are restricted to a certain class, $\mathcal{L}$ defined below. To define this class we need to consider the halting set over the halting set : $\emptyset''=(\emptyset')'$ (See\cite{misc}). Let $\mathcal{L}$ denote the union of the two classes of states.
\begin{enumerate}
    \item States in $A^{\infty}_B$ for some $B$, as in Corollary \ref{cor}
    \item States which do not Turing compute $\emptyset''$.
    \end{enumerate}
 Nies and Barmpalias (in personal communication) have shown that q-MLR is equivalent to weak quantum Solovay randomness for states which do not compute $\emptyset''$. The same equivalence also holds on $A^{\infty}_B$ by Corollary \ref{cor}. This similarity motivates our study of $\mathcal{L}$. 
\begin{thm}
\label{thm:39}
If $\rho=(\rho_{n})_{n} \in  \mathcal{L}$, then the following are equivalent
\begin{enumerate}
\item $\rho$ is q-MLR.
    \item $\rho$ is weak Solovay random.
    \item $ \forall \epsilon>0 \forall c >0 \forall^{\infty} n,  QK^{\epsilon}(\rho_{n}) \geq n+c.$
    \item $\forall \epsilon>0\exists c \forall n, QK^{\epsilon}(\rho_{n})>n-c .$
\end{enumerate}
\end{thm}
\begin{proof}
(1)$\iff$(2) follows from the previous remarks.\\
(4)$\Longrightarrow$(1):\begin{proof} First, let $\rho \in A^{\infty}_B$ for some $B$ and let (4) hold. By the same argument as in Corollary \ref{cor}, we get that
\[  \exists c >0 \forall  n  ,K(\rho'\upharpoonright n) \geq n-c .\]
The classical Levin\textendash Schnorr result \cite{misc1} implies that $\rho'$ is MLR. Using once more 3.13 in \cite{unpublished} as in \ref{cor}, we see that $\rho$ is q-MLR. Now suppose $\rho$ does not Turing compute $\emptyset''$. We will show that (4) implies (2). Suppose for a contradiction that $(S^{m})_{m}$ is a strong Solovay test which $\rho$ fails at $\epsilon'>0$. By Theorem 2.11 in \cite{bhojraj2020quantum}, we can effectively compute a q-MLT $(G^{m})_{m}$ which $\rho$ fails at some rational $\epsilon>0$. Let $g(m):=$ the least $s$ such that Tr$(\rho_sG^{2m}_{s})>\epsilon$. As $\rho$ computes $g$, by Martin's high domination theorem (see \cite{misc1} for a proof), there is a total computable function $f$ such that $\exists^{\infty} g(n)<f(n)$. We may assume that $f(t)>3t$ for all $t$ by taking the max of 2 computable functions. Fix this $f$ (non-uniformly) and consider the following machine, $M$:\\ On input $0^{m}1$, $M$ outputs $F^{m}$ where $F^{m}$ is such that
\[G^{2m}_{f(m)}= \sum_{v \in F^{m}}|v\big>\big<v|.\]
Clearly $M$ is prefix free. Let $l-1$ be it's coding constant. Let $t$ be so that $f(t)>g(t)$. Let $F^t$ be defined similarly to $F^m$ above. Then, by definition of $g$, we have that \[\epsilon< \sum_{v \in F^{t}}\big<v|\rho_{f(t)}|v\big>.\] $M(0^{t}1)=F^{t}$ and so, there is a $\pi$ such that $|\pi|\leq t+l $ and $\mathbb{U}(\pi)=F^{t}$. Also note that $|F^{t}|\leq 2^{f(t)-2t}$ by the definition of a q-MLT. So,
$
\text{QK}^{\epsilon}(\rho_{f(t)}) \leq t+l +f(t)-2t=f(t)-t+l.$

Recall that $t$ was an arbitrary element of the infinite set $\{s: f(s)>g(s)\}$. So, for infinitely many $t$s, there is an $n=f(t)$ such that $\text{QK}^{\epsilon}(\rho_{n}) \leq n-t+l,$  contradicting $(4)$.
\end{proof}
$(3)\Longrightarrow(4)$ is obvious and $(2)\Longrightarrow(3)$ was done in Theorem \ref{thm:7}.
\end{proof}
We apply the preceding theorem to get the following quantum analog of a classical result, Proposition 3.2.14 in \cite{misc}.
\begin{thm}
\label{thm:compset}
Let $C$ be an infinite computable set, $\rho \in \mathcal{L}$ and $\epsilon>0$. If there is a $d$ such that for all $m\in C$, $QK^{\epsilon}(\rho_m)>m-d$, then $\rho$ is weak Solovay random.
\end{thm}
\begin{proof}
Let $M$ be the machine doing the following: On input $\sigma$, check if $\mathbb{U}(\sigma)=F$, an orthonormal set $F\subseteq \mathbb{C}^{2^{n}}$. If such a $F$ and $n$ exist, compute $s$ such that $n+s$ is the least element of $C$ greater than $n$ and output the set:
\[T:= \{v \otimes \pi: v \in F, \pi \in 2^{s}\}.\]
Note that $|T|=2^{s}|F|$. It is easy to see that $M$ is prefix-free. Let $l$ be it's coding constant. Suppose for a contradiction that $\rho$ is not weak Solovay random. \ref{thm:39} implies that $\forall c$, $\exists n_{c}$ such that $QK^{\epsilon}(\rho_{n_{c}}) \leq n_{c}-c.$ Let $c$ be arbitrary and take such an $n:=n_c$. There is a $\sigma$ and $F$ such that $\mathbb{U}(\sigma)=F \in \mathbb{C}^{2^{n}}$, $|\sigma|+$log$(|F|)\leq n-c$ and \[\sum_{v\in F} \big<v|\rho_{n}|v\big> >\epsilon.\] Let $t=n+s$ be the least element of $C$ greater than $n$. On input $\sigma$, $M$ outputs $T$ as above.
Note that \[Q:=\sum_{w\in T}|w\big>\big<w|= \sum_{v\in F, \pi\in 2^{s}}|v\big>\big<v|\otimes |\pi\big>\big<\pi|=\big(\sum_{v\in F }|v\big>\big<v|\big)\otimes \big( \sum_{\pi\in 2^{s}}|\pi\big>\big<\pi|\big)=W\otimes I,
\]
where $W:=\sum_{v\in F}|v\big>\big<v|$ and $I$ be the identity on $\mathbb{C}^{2^{s}}$.
Then, by the coherence property of states, \[\sum_{w\in T} \big<w|\rho_{t}|w\big>=\text{Tr}(\rho_{t}Q)=\text{Tr}(\rho_{t}[W\otimes I])=\text{Tr}(\rho_{n}W)>\epsilon.\]
Consequently, \[QK^{\epsilon}(\rho_{t})\leq |\sigma|+  \text{log}(|T|)+ l = |\sigma|+  \text{log}(|F|)+s+ l\leq n-c+s+l=t-c+l.\]
Since $d$ and $l$ were constants and $c$ was arbitrary, this contradicts the assumption.
\end{proof}

\subsection{A weak Levin\textendash Schnorr type result}
\label{subsect:weak}
Theorem \ref{thm:7} implies that if $\rho$ is weak-Solovay random then, $ \forall \epsilon>0\exists c \forall n, QK^{\epsilon}(\rho_{n})>n-c$. I.e., being strong-Solovay random implies the Levin\textendash Schnorr  condition. Does this reverse? We give two partial results in this direction: the Levin\textendash Schnorr condition implies that $\rho$ passes all strong-Solovay tests of a certain type.
\begin{defn}
For a rational $s\in(0,1)$, a $s$-strong Solovay test is a strong Solovay test $(S^{r})_{r}$ such that
$\sum_{r}\tau(S^{r})^{s} < \infty $ and $\sum_{r}\tau(S^{r}) $ is a computable real number.
\end{defn}

\begin{thm}
\label{thm:mkssr}
If $ \forall \epsilon>0\exists c \forall n, QK^{\epsilon}(\rho_{n})>n-c$, then $\rho$ passes all $s$-strong Solovay tests for all rational $s\in (0,1)$.
\end{thm}

\begin{proof}
Suppose for a contradiction that $(S^{m})_{m}$ is a $s$-strong Solovay test which $\rho$ fails at $\epsilon>0$ and $\sum_{i} \tau(S^{i})=Q $, computable. For all $m$, let $S^m$ be $2^{n_{m}}$ by $2^{n_{m}}$ and we may let the $n_m$s be distinct. Let $f(m):=-$ log$(\tau(S^m))$ and $g(m):=\lceil s f(m) \rceil$.
Partition $\omega$ into the fibers induced by $g$. ($P$ is a fiber of $g$ if $P=g^{-1}(\{x\})=\{y: g(y)=x\},$ for some $x$.). Note that $g(r)\geq  -s \text{log}\tau(S^{r})$,
and hence
$2^{-g(r)}\leq \tau(S^{r})^{s} .$ In particular, this implies that $P$ is finite. So, there are  countably infinitely many fibers,  $\{P_{1},P_{2}\dots\}$ and $\omega = \bigcup_{m} P_{m},$
where for all $m$, there is an $x_m$ such that $P_{m}=g^{-1}(\{x_{m}\})$ and $m \mapsto x_{m}$ is injective.

The fiber $P$ of $x=g(z)$ can be computed from $x$ as follows. Note that $g(c)=x$ iff, $f(c) \in [s^{-1}(x-1), s^{-1}x]$. As $Q$ is computable, compute an interval $J$ such that, $|J|<2^{-s^{-1}x}$,  $Q\in J$ and $\sum_{r\leq q} 2^{-f(r)}\in J $ for some $q$. So, $f(c)>s^{-1}x$ if $c>q$. $P$ can be computed by evaluating $g$ on $[0,q]$.

The idea is to describe $S^r$ by computing the fiber $P$ containing $r$ and then specifying the location of $r$ in the lexicographical ordering on $P$. As $(S^r)_r$ is a $s$-strong Solovay test, this description of $S^r$ is short enough to derive a contradiction. Consider the  machine, $M$ doing the following:  On input $\lambda$, check if there is a decomposition $\lambda=\pi\sigma$  such that,
\begin{itemize}
    \item There is an $m$ such that $\mathbb{U}(\pi)=0^{g(m)}$.
    \item $|\sigma|=t=\lceil\text{log}(|P|)\rceil$ where  $P$ is the fiber of $g(m)$. (Recall that $P$ can be computed from  $g(m)$.) 
    
\end{itemize}

 If these hold, then order $P$ lexicographically using the ordering on $2^t$ and let $r$ be the $\sigma^{th}$ element in this ordering. Output $F^{r}$ where $F^{r}$ is such that $S^{r}= \sum_{v \in F^{r}}|v\big>\big<v|.$
 
Note that $M$ is prefix free: Suppose $\lambda,\lambda' \in $ dom$(M)$ as witnessed by $\lambda=\pi\sigma$ and $\lambda=\pi'\sigma'$. So, $M$ finds $m$ and $m'$ such that $\mathbb{U}(\pi)=0^{g(m)}$ and $\mathbb{U}(\pi')=0^{g(m')}$. Let $\pi\sigma \preceq \pi'\sigma'$. Then, it must be that $\pi\preceq\pi'$ or $\pi'\preceq\pi$. Since $\mathbb{U}$ is prefix free, it follows that $\pi=\pi'$. So, $g(m)=g(m')$. Hence, $m$ and $m'$ are in the same fiber, $P$. Letting $t= \lceil \text{log}(|P|)\rceil$, $\lambda,\lambda' \in $ dom$(M)$ implies that $|\sigma|=|\sigma'|=t$.

For each $m\in \omega$, let $r_m$ be any element from $P_m$. Then,
\begin{align}
\label{eq:mkssr3}
    \sum_{m}|P_{m}|2^{-g(r_{m})} = \sum_{m}\sum_{r\in P_{m}} 2^{-g(r)} \leq \sum_{m}\sum_{r\in P_{m}} \tau(S^{r})^{s}= \sum_{i} \tau(S^{i})^{s}< \infty.
\end{align}

Let $h$ be the function defined by, $h(m):=$ log$(|P_{m}|)-g(r_{m})$ where $r_m$ is any representative from $P_m$. By \eqref{eq:mkssr3},
$|P_{m}|2^{-g(r_{m})} \rightarrow 0$ as $m\rightarrow \infty$. So, $h(m) \rightarrow -\infty$ as $m\rightarrow \infty$.
Each fiber is finite and $\rho$ fails the test at $\epsilon$. So, there is an infinite set $I= \{P_{j_{1}}, P_{j_{2}},\dots\}$ such that for all $i$, there is a $t_{i}\in P_{j_{i}}$ with  
Tr$(\rho_{n_{t_{i}}}S^{t_{i}})>\epsilon$.  $j_{i} \rightarrow  \infty$ as $i\rightarrow \infty$ and so, $h(j_{i})= $ log$(|P_{j_{i}}|)-g(r_{t_{i}}) \rightarrow -\infty$ as $i\rightarrow \infty$. This asymptotic behavior will be used below to derive a contradiction.

Fix an arbitrary $i$ and a $t=t_{i}\in P_{j_{i}}$ as above.
So, $
\epsilon< \sum_{v \in F^{t}}\big<v|\rho_{n_{t}}|v\big>.$

Let $t$ be the $\sigma^{th}$ element of $P_{j_{i}}$ in the lexicographic ordering used by $M$ and let $\mathbb{U}(\pi)=0^{g(t)}$ and $K(0^{g(t)})=|\pi|$. Then,
$M(\pi\sigma)=F^{t}$ and so, there is a bitstring $\iota$ such that $|\iota|\leq^{+} K(0^{g(t)})+ \lceil \text{log}(|P_{j_{i}}|)\rceil $ and $\mathbb{U}(\iota)=F^{t}$. Note that log$|F^{t}|=$ log(Tr($S^t))=$ log$(\tau(S^{t})) + n_{t}= -f(t) + n_{t}$. Let $d:=(1-s)s^{-1} >0$. So, $\{(\lceil nd \rceil, 0^{n}): n \in \omega\},$
is a bounded request set (see \cite{misc1} for a definition) and hence $K(0^{n})\leq^{+}\lceil nd \rceil$. Using all this, we get that:
\begin{align*}
\text{QK}^{\epsilon}(\rho_{n_{t}})
&\leq^{+} K(0^{g(t)})+\lceil \text{log}(|P_{j_{i}}|)\rceil  -f(t) +  n_{t}\\
&\leq^{+} \lceil d g(t)\rceil+\lceil \text{log}(|P_{j_{i}}|)\rceil -f(t) +  n_{t}\\
&\leq^{+} dsf(t)+\lceil \text{log}(|P_{j_{i}}|)\rceil -f(t) +  n_{t}\\
&=f(t)(ds-1)+\lceil \text{log}(|P_{j_{i}}|)\rceil  +  n_{t}\\
&=-sf(t) +\lceil \text{log}(|P_{j_{i}}|)\rceil  +  n_{t}\\
&\leq^{+}-g(t)+\lceil \text{log}(|P_{j_{i}}|)\rceil  +  n_{t}\\
&=h(j_{i})+n_{t}.
\end{align*}
The last equality follows as $t=t_i$ is in  $P_{j_{i}}$. This means that there is an infinite sequence $(n_{t_{i}})_{i}$ such that
\[\text{QK}^{\epsilon}(\rho_{n_{t_{i}}}) <^{+} n_{t_{i}}+h(j_{i}) .\]
Finally, recall that $h(j_{i})\rightarrow -\infty$ as $i\rightarrow \infty$ and we have a contradiction.
\end{proof}
Theorem \ref{thm:mkssr} can be strengthened by weakening the defining criteria for a $s$-strong Solovay test.
\begin{defn}
\label{def:phi}
Let $s\in(0,1)$ be a rational. Let $\phi$ be any computable, non-decreasing, non-negative function on the reals such that $\phi(sr)\geq^{\times} s\phi(r)$ (I.e., there is a $C>0$ \emph{independent of s}, such that for all $r$,  $C\phi(sr)\geq  s\phi(r)$) and $\sum_{n}2^{-n}\phi(n)<\infty$ (So, $\phi$ does not tend to infinity too fast). A $(\phi,s)$-strong Solovay test is a strong Solovay test $(S^{r})_{r}$ such that
\begin{align}
\label{eq:mkssr6}
    \sum_{r}\dfrac{\tau(S^{r})^{s}}{\phi(-\text{log}(\tau(S^{r})))} < \infty ,
\end{align}

and
\[\sum_{r}\tau(S^{r}) = Q,\] where $Q$ is a computable real number.
\end{defn}
The term in the denominator in \eqref{eq:mkssr6} tends to infinity with $r$ and hence it is easier for a strong Solovay test to be a $(\phi,s)$-strong Solovay test than to be a $s$-strong Solovay test. So, passing all $(\phi,s)$-strong Solovay tests is a more restrictive notion of randomness than passing all $s$-strong Solovay tests. So, the following theorem is an improvement of, and implies Theorem \ref{thm:mkssr}.
\begin{thm}
If $ \forall \epsilon>0\exists c \forall n, QK^{\epsilon}(\rho_{n})>n-c$, then $\rho$ passes all $(\phi,s)$-strong Solovay tests for all rational $s\in (0,1)$ and all $\phi$ as in Definition \ref{def:phi}.
\end{thm}

\begin{proof}
Suppose for a contradiction that  $(S^{m})_{m}$ is a $(\phi,s)$-strong Solovay test which $\rho$ fails at $\epsilon>0$ and $\sum_{i} \tau(S^{i})=Q $, computable. For all $m$, let $S^m$ be $2^{n_{m}}$ by $2^{n_{m}}$ and we may let the $n_m$s be distinct.
For ease of presentation, we do the proof in 2 cases. First, let $s\leq 0.5$.
Let $f(m):=-$ log$(\tau(S^m))$ and let $g(m):=\lceil s f(m) \rceil$.
Partition $\omega$ into the fibers induced by $g$.
Fix some fiber $P$ of some $x=g(r)$. I.e., $r$ is a representative from $P$. Then, $g(r)=\lceil s(- \text{log}\tau(S^{r})) \rceil$. So, $g(r)\geq  -s \text{log}\tau(S^{r}),$
and hence
$2^{-g(r)}\leq \tau(S^{r})^{s} .$
In particular, this implies that each fiber is finite. So, there are  countably infinitely many fibers,  $\{P_{1},P_{2}\dots\}$. So, $\omega = \bigcup_{m} P_{m},$
where for all $m$, there is an $x_m$ such that $P_{m}=g^{-1}(\{x_{m}\})$ and $m \mapsto x_{m}$ is injective. 
For each $m\in \omega$, let $r_m$ be any representative from $P_m$.
 For all $r$, $sf(r)\leq g(r)$ and $\phi$ is non-decreasing. So,
\begin{align}
\label{eq:mkssr13}\sum_{m} |P_{m}|\dfrac{2^{-g(r_{m})}}{\phi(g(r_{m}))} = \sum_{m}\sum_{r\in P_{m}} \dfrac{2^{-g(r)}}{\phi(g(r))}\leq  \sum_{m}\sum_{r\in P_{m}} \dfrac{\tau(S^{r})^{s}}{\phi(sf(r))} \leq^{\times} \sum_{r} \dfrac{\tau(S^{r})^{s}}{s\phi(f(r))}< \infty.
\end{align}

The fiber $P$ of $x=g(z)$ can be computed from $x$ for the same reason as in the previous proof. Its idea of `compressing' $S^r$ is also used here. 
\\
Consider the  machine, $M$ doing the following:  On input $\lambda$, search for a decomposition $\lambda=\pi\sigma$, such that
\begin{itemize}
    \item $\mathbb{U}(\pi)=0^{g(m)}$ for some $m$.
    \item
    $|\sigma| = t$ where
    $P$ is the fiber containing $m$, (which can be computed from $g(m)$) and $t= \lceil\text{log}(|P|)\rceil$.
\end{itemize}
If found, order $P$ lexicographically using the ordering on $2^t$ and let $r$ be the $\sigma^{th}$ element in this ordering. Output $F^{r}$ where $F^{r}$ is such that
$S^{r}= \sum_{v \in F^{r}}|v\big>\big<v|.$
Note that $M$ is prefix free for the same reason as in the previous proof. Let $l$ be $M$'s coding constant.
Let $h$ be the function defined by, $h(m):=$ log$(|P_{m}|)-g(r_{m})-$log $\phi g(r_{m})$, where $r_m$ is any representative from $P_m$. By \eqref{eq:mkssr13},
\[\dfrac{|P_{m}|2^{-g(r_{m})}}{\phi g(r_{m})} \rightarrow 0\] as $m\rightarrow \infty$. So, $h(m) \rightarrow -\infty$ as $m\rightarrow \infty$.
Each fiber is finite and $\rho$ fails the test at $\epsilon$. So, there is an infinite set $I= \{P_{j_{1}}, P_{j_{2}},\dots\}$ such that for all $i$, there is a $t_{i}\in P_{j_{i}}$ with  
Tr$(\rho_{n_{t_{i}}}S^{t_{i}})>\epsilon$.  $j_{i} \rightarrow  \infty$ as $i\rightarrow \infty$ and so, $h(j_{i}) \rightarrow -\infty$ as $i\rightarrow \infty$. This asymptotic behavior will be used below to derive a contradiction.

Fix an arbitrary $i$ and a $t=t_{i}\in P_{j_{i}}$ as above.
So,
\begin{align} 
\label{eqn:mkssr1}
\epsilon< \sum_{v \in F^{t}}\big<v|\rho_{n_{t}}|v\big>.
\end{align}

Let $t$ be the $\sigma^{th}$ element of $P_{j_{i}}$ in the lexicographic ordering used by $M$. Let $\pi$ be such that $K(0^{g(t)})=|\pi|$ and $\mathbb{U}(\pi)=0^{g(t)}$. Then,
$M(\pi\sigma)=F^{t}$ and so, there is a bitstring $\kappa$ such that $|\kappa|\leq K(0^{g(t)}) + \lceil \text{log}(|P_{j_{i}}|)\rceil + l$ and $\mathbb{U}(\kappa)=F^{t}$. Note that log$|F^{t}|=$ log(Tr($S^t))=$ log$(\tau(S^{t})) + n_{t}= -f(t) + n_{t}$. So,
\[
\text{QK}^{\epsilon}(\rho_{n_{t}}) \leq K(0^{g(t)})+\lceil \text{log}(|P_{j_{i}}|)\rceil +l -f(t) +  n_{t}.
\]
  
Note that $\{(n-\lceil\text{log}(\phi(n))\rceil,0^{n}): n\in \omega\}$
is a bounded request set by definition of a $(\phi,s)$ test and so
$K(0^{g(t)})\leq^{+} g(t)-$log$(\phi g(t))$. So,
\[
\text{QK}^{\epsilon}(\rho_{n_{t}}) \leq^{+} g(t)-\text{log}(\phi g(t))+ \lceil \text{log}(|P_{j_{i}}|)\rceil  -f(t) +  n_{t}.
\]
Since
$g(t)-1 < sf(t) $, we have
\begin{align}
    \label{eq:mkssr12}
    -g(t)+1>-sf(t).
\end{align}

Since $s\leq 0.5$ we have that $1-s\geq s$.\\ So,
$g(t)-f(t) \leq sf(t)-f(t)+1 = -(1-s)f(t)+1\leq -sf(t)+1.$ 
Using \eqref{eq:mkssr12}, 
\[g(t)-f(t)<-g(t)+2.\]
So,
\[
\text{QK}^{\epsilon}(\rho_{n_{t}}) <^{+} -g(t)+\lceil \text{log}(|P_{j_{i}}|)\rceil -\text{log}(\phi  g(t)) +  n_{t} = h(j_{i}) +  n_{t}.
\]
The equality follows as $t=t_i$ is in  $P_{j_{i}}$. This means that there is an infinite sequence $(n_{t_{i}})_{i}$ such that
$\text{QK}^{\epsilon}(\rho_{n_{t_{i}}}) <^{+} n_{t_{i}}+h(j_{i}).$
Finally, recall that $h(j_{i})\rightarrow -\infty$ as $i\rightarrow \infty$ and we have a contradiction.

Now let $s>0.5$ and let $f,g$ be as in the previous case. Let $b(m):= \lceil (1-s)f(m)\rceil$ and let $C:= \lceil s/(1-s)\rceil + 1$.

Consider the machine $M$ doing the following: on
input $\pi1^{y}0\sigma$, check if the following conditions hold.
\begin{itemize}
    \item 
    There is $m$ such that $\mathbb{U}(\pi)=0^{b(m)}$.
    \item If $x=b(m)$,
     $J=(s(1-s)^{-1}(x-1),s(1-s)^{-1}x + 1] \cap \omega$ and $w$ is the $y^{th}$ element of $J$, then there is a $z$ such
     that $g(z)=w$.
    
     \item
     If $P$ is the fiber of $g$ containing $z$ (P is computable from $w=g(z)$ just as in the previous case) and $t= \lceil\text{log}(|P|)
     \rceil$, then $|\sigma| = t$
     \end{itemize}
     
     If all the above are met, then order $P$ lexicographically using the ordering on $2^t$ and let $r$ be the $\sigma^{th}$ element in this ordering. Output $F^r$.

Roughly, the idea is as follows: Just as in the previous case, we want to compress $F^r$ where $r$ is the $\sigma^{th}$ number in the fiber of $w=g(m)$. The first step to achieve this is to describe $g(m)$. While in the previous case we used an $\iota$ such that $\mathbb{U}(\iota)=0^{g(m)}$ and $|\iota|=K(0^{g(m)})$), we use here the shorter string $\pi$ where  $\mathbb{U}(\pi)=0^{b(m)}$ and $|\pi|=K(0^{b(m)})$  together with $1^y$ for describing $g(m)$. From $\pi$, we get $x=b(m)$ which in turn gives $J$ which contains $g(m)$. So, $\pi$ along with $y$, the location of $g(m)$ in $J$, describes $g(m)$. After $g(m)$ is found, $F^r$ can be described just as in the previous case. The details are: As $x=b(m)$, $(1-s)f(m)\in (x-1,x]$ and hence $sf(m)$ lies in $(s(1-s)^{-1}(x-1),s(1-s)^{-1}x]$.
So, $\lceil sf(m) \rceil =  g(m)\in J=  (s(1-s)^{-1}(x-1),s(1-s)^{-1}x + 1]  \cap \omega$. Since $|J|\leq C$, $g(m)\in J$ can be determined by specifying $y \leq C$, it's location in $J$. So, $M$ can recover $g(m)$. From this point on, the remaining procedure is the same as in the previous case.
\\
We see that $M$ is prefix-free: Let $\pi1^{y}0\sigma$ and $\pi'1^{y'}0\sigma'$ be in the domain of $M$ and let $\pi1^{y}0\sigma\preceq \pi'1^{y'}0\sigma'$. By the same argument as in case1, $\pi=\pi'$ and  $M$ finds some $m,m'$ with $x=b(m)=b(m')$. It follows that $y=y'$. So, if $w$ is the $y^{th}$ (and  $y'^{th}$)  element of $J$ (as above), then $M$ finds $z,z'$ such that $g(z)=w=g(z')$. So, $z$ and $z'$ are in the same fiber $P$ and it hence follows as in the previous case that $|\sigma|=|\sigma'|$. Define $I$ and $h$ exactly as in the previous case. Fix some $i$ and let $t=t_i$ be an element of $P_{j_{i}}$ such that \eqref{eqn:mkssr1} holds. Let $t$ be the $\sigma^{th}$ element of $P_{j_{i}}$. Let $x=b(t)=\lceil (1-s)f(t)  \rceil$. So, $(1-s)f(t) \in (x-1,x]$ and  $sf(t) \in (s(1-s)^{-1}(x-1),s(1-s)^{-1}x]$. Hence, $g(t) \in (s(1-s)^{-1}(x-1),s(1-s)^{-1}x+1] \cap \omega=J$ and let $g(t)$ be the $y^{th}$ element of $J$.
Let $\pi$ be such that   $\mathbb{U}(\pi)=0^{b(t)}$ and $|\pi|=K(0^{b(t)})$. Then, on input $\pi1^{y}0\sigma$, M finds some $z$ (it could be that $z=t$, but not necessarily) such that $b(t)=b(z)=x$ and then finds that the $y^{th}$ element of $J$ is $g(z')$ for some $z'$ (again, although  $g(z')=g(t)$, it could be that $t=z'$ but not necessarily). Since $z'$ and $t$ are both in $P_{j_{i}}$, $M$ outputs $F^{t}$ after reading $\sigma$. So, there is a $\pi$
 such that $\mathbb{U}(\pi)=F^{t}$ and
 $|\pi|\leq^{+} K(0^{b(t)})+C+\lceil\text{log}(|P_{j_{i}}|)\rceil $.\\
 So, $ \text{QK}^{\epsilon}(\rho_{n_{t}})$ 
\begin{align*}
&<^{+}K(0^{b(t)})+ \lceil\text
 {log}(|P_{j_{i}}|)\rceil + \text{log}(|F^{t}|)\\
 &\leq^{+}b(t)-\lceil\text{log}(\phi(g(t)))\rceil+\lceil\text{log}(|P_{j_{i}}|)\rceil + n_{t}-f(t)\\
 &\leq^{+}-g(t)-\lceil\text{log}(\phi(g(t)))\rceil+\lceil\text{log}(|P_{j_{i}}|)\rceil + n_{t} 
\end{align*}

The last inequality is since, by \eqref{eq:mkssr12} (which holds for any $s$), $b(t)-f(t) \leq (1-s)f(t)-f(t)+1=-sf(t)+1<-g(t)+2$. Since $t=t_{i} \in P_{j_{i}}$, we see that  $ \text{QK}^{\epsilon}(\rho_{n_{t}}) <^{+}h(j_{i})-n_{t_{i}}.$ for all $i$. This gives a contradiction for the same reason as in the previous case.

\end{proof}

\subsection{QK and computable measure machines}
\label{schn}
Schnorr randomness is an important randomness notion in the classical realm\cite{misc,misc1}. While $K$ plays well with Solovay randomness, $K_C$, a version of $K$ using a computable measure machine, $C$ (a prefix-free Turing machine whose domain has computable Lebesgue measure) gives a Levin-Schnorr characterization of Schnorr randomness  (See theorem 7.1.15 in \cite{misc1}). 

So, with the intention of connecting it to quantum Schnorr randomness, we define $QK_C$ a version of $QK$ using a computable measure machine, $C$.

Theorem \ref{thm:677} shows that $QK_C$ agrees with $K_C$ on the classical bitstrings. Analogously to the classical case, Theorem \ref{thm:schnor} is a Levin\textendash Schnorr type of characterizations of quantum Schnorr randomness using $QK_C$. Theorem \ref{thm:Chaitin}, a Chaitin type characterization of quantum Schnorr randomness using $QK_C$ implies Theorem \ref{class}, a Chaitin type characterization of \emph{classical} Schnorr randomness in terms of $K_C$.

For $C$ a computable measure machine and $\sigma$ a string, $K_C$ is defined analogously to $K$; $K_C(\sigma):=$inf$\{|\tau|: C(\tau)=\sigma\}.$ The quantum version is:
for $C$, a computable measure machine and a $\epsilon>0$, define 
$QK_{C}^{\epsilon}(\tau)$ to be:
\begin{defn}
\label{defn:4}
$QK_{C}^{\epsilon}(\tau) := $ inf  $\{|\sigma|+$log$|F|: C(\sigma)\downarrow = F$, a orthonormal set in $\mathbb{C}^{2^{|\tau|}}_{alg}$ and $\sum_{v \in F} \big<v |\tau |v \big> > \epsilon \}$

\end{defn}
The infimum of the empty set is taken to be $\infty$.
Notation: In this section, $\mu$ denotes Lebesgue measure and $C_t$ denotes $C$ run upto the $t$ steps. We may assume that dom$(C_t)\subseteq 2^t$. By a \emph{sequence}, we mean a countable collection whose elements may possibly be repeated. If $S$ is a sequence, the sum $\sum_{s\in S}$ will be over all elements of $S$, with repetition.

Similarly to Theorem \ref{thm:67}, we show that $QK_C$ `agrees with' $K_C$ on the classical qubitstrings. In Theorem \ref{thm:677} and its proof, $P$ and $C$ will stand for computable measure machines.
\begin{thm}
\label{thm:677}
 For all rational $\epsilon>0$ and all $C$, there exists a $P$ such that $K_P(\sigma) \leq  QK^{\epsilon}_C(|\sigma\big>\big<\sigma|)+1$ for all classical bitstrings $\sigma$.
\end{thm}

\begin{proof}
The proof is almost identical to that of Theorem \ref{thm:67}. Fix a rational $\epsilon>0$ and a $C$. Consider the machine $P$ from the proof of Theorem \ref{thm:67} but with $\mathbb{U}$ replaced by $C$. We now show that $\mu($dom$(P))$ is computable. Let $\delta>0$ be arbitrary. Since $\mu($dom$(C))$ is computable, find a stage $t$ so that $\mu(\text{dom}(C))-\mu \text{(dom}( C_t )) < \delta.$
(The $t$ can be found as follows: Compute a $q'$ such that $|q-q'|< \delta/2$. So, $q'-\delta/2<q<q'+ \delta/2$. Since $\mu \text{(dom}( C_t ))\nearrow q$, as $t\rightarrow \infty$, we can compute a $t$ such that, $q'-\delta/2<\mu \text{(dom}( C_t ))<q'+ \delta/2$.). We may compute $S$, the set of those strings $\pi=\sigma\tau \in$ dom$(P)$ and $\sigma \in $ dom$(C_t)$. So, dom$(P)\backslash S$ consists of strings $\pi=\sigma\tau$ such that $\sigma\in$ dom$(C)\backslash$dom$(C_t)$. So, it is easy to see that $\mu($dom$(P))$-$\mu(S) < \mu($dom$(C))$-$\mu($dom$(C_t)) <\delta$. As $\delta>0$ was arbitrary, this shows that $\mu($dom$(P))$ is computable.
Now, let $\sigma \in 2^n$ be any classical bitstring such that $QK_C^{\epsilon}(|\sigma\big>\big<\sigma|)<\infty$. Let $\lambda$ and $F \subseteq \mathbb{C}^{2^{n}}_{alg}$ orthonormal such that $|\lambda|+ \text{log}(|F|)=QK_C^{\epsilon}(|\sigma\big>\big<\sigma|)$,  $\sum_{v \in F} \big<v |\sigma\big>\big<\sigma|v \big> > \epsilon$ and $C(\lambda)=F$. Let $O:= \sum_{v\in F} |v\big>\big<v|$. Note that since
$\epsilon<\sum_{v \in F} \big<v |\sigma\big>\big<\sigma|v \big> = \big<\sigma|O|\sigma\big>$, $\sigma \in S^{\epsilon}_{E,O}$ where $E$ is the standard basis. Let $\tau$ be a length $ \lceil \text{log}(\epsilon^{-1}|F|)\rceil$ string such that $g(\tau)=\sigma$. Then, we see that $P(\lambda \tau) = \sigma $. So, \[K_P(\sigma)\leq  |\lambda|+|\tau|\leq |\lambda|+\text{log}(\epsilon^{-1})+1+ \text{log}(|F|) = QK^{\epsilon}_C(|\sigma\big>\big<\sigma|)+1.\]

\end{proof}
\begin{remark}
\label{rem:qkc}
Theorem \ref{thm:677} establishes one direction of the coincidence of $K_C$ and $QK_C$ for classical qubitstrings. In the other direction, take some classical bitstring $\sigma$ with $K_C(\sigma)<\infty$ for some $C$. Let $C(\pi)=\sigma$ and $|\pi|=\text K_C(\sigma)$. Then, letting $F=\{\sigma\}$ in \ref{defn:4}, $QK_C^{\epsilon}(|\sigma\big>\big<\sigma|)\leq QK_C^{1}(|\sigma\big>\big<\sigma|)\leq |\pi|=K_C(\sigma)$, for any $\epsilon>0$. 
\end{remark}
\subsection{Quantum Schnorr randomness and $QK_C$}
Theorem \ref{thm:schnor} is a quantum analogue of the classical characterization of Schnorr randomness: $X$ is Schnorr random if and only if for any computable measure machine, $C$, there is a constant $d$ such that for all $n$, $K_{C}(X\upharpoonright n)>n-d$.
\begin{thm}
\label{thm:schnor}
A state $\rho$ is quantum Schnorr random if and only if for any computable measure machine, $C$ and any $\epsilon > 0$, there is a constant $d>0$ such that for all $n$, $QK^{\epsilon}_{C}(\rho_{n})>n-d$.
\end{thm}
\begin{proof}
($\Rightarrow$) We prove it by contraposition. I.e., show that $\rho$ is not quantum Schnorr random if there is a $C$ and an $\epsilon>0$ such that for all $d$, there is an $n=n_d$ such that $QK^{\epsilon}_{C}(\rho_{n}) \leq n-d$. Let $T_{s}$ be the set of all $\sigma$ such that $C_{s}(\sigma)\downarrow = F_{\sigma}$, an orthonormal set such that $|\sigma|+$ log $|F_{\sigma}|< n_{\sigma} $ and $F_{\sigma} \subseteq  \mathbb{C}^{2^{n_{\sigma}}}$ for some $n_{\sigma}$. Let $T=\bigcup_{s}T_{s}.$ For all strings $\sigma$, let $P_{\sigma}:= \sum_{v\in F_{\sigma}} |v \big>\big<v |.$
Let $Q_s$ be the sequence of those $P_{\sigma}$ for $\sigma \in T_s$, $Q$ the sequence of those $P_{\sigma}$ for $\sigma \in T$ and $D_s$ the sequence of those $P_{\sigma}$ for $\sigma \in T\backslash T_s$.
Next, we show that $\alpha:=\sum_{P\in Q} \tau(P)$ is computable by showing how to approximate it within $2^{-k}$ for an arbitrary $k$:  Computably find a $t$ (using the same method as in Theorem \ref{thm:677}) such that $
    \mu(\text{dom}(C )) -\mu \text{(dom}(C_t ))  < 2^{-k}.$
 We show that  $\sum_{P_{\sigma}\in Q_{t}}\tau(P_{\sigma})$ is within $2^{-k}$ of $\alpha$. Note that for all $\sigma \in T$, $2^{|\sigma|}|F_{\sigma}|< 2^{n_{\sigma}}$. So, $ \tau(P_{\sigma})= 2^{-n_{\sigma}}|F_{\sigma}|< 2^{-|\sigma| }$. 
\[\alpha -\sum_{P_{\sigma}\in Q_{t}} \tau(P_{\sigma})  = \sum_{P_{\sigma}\in D_t} \tau(P_{\sigma}) =
     \sum_{\sigma \in T/T_{t}}|F_{\sigma}|2^{-n_{\sigma}} 
     \leq\sum_{\sigma\in T/T_{t}}2^{-|\sigma|}\]\[
     \leq \sum_{\sigma\in\text{dom}( C)/\text{dom}( C_t ) }2^{-|\sigma|}
     \leq  (\mu(\text{dom}( C ))-\mu (\text{dom}( C_t )))< 2^{-k}\]

Note that $\sum_{P_{\sigma}\in Q_{t}}\tau(P_{\sigma})$ is a rational, uniformly computable in $t$ since dom$(C_t)\subseteq 2^t$ is uniformly computable in $t$. This shows that $Q$ is a quantum Schnorr test. By the assumption, we see that is a infinite sequence $d_1 < d_2 < \cdots$ and a list of distinct natural numbers $n_{d_{1}}, n_{d_{2}} \cdots$ so that for all $i$, there is a $P_i$ in $Q$ such that Tr$(P_i\rho_{n_{d_{i}}})>\epsilon$. So, $\rho$ fails $Q$ at $\epsilon$.

($\Leftarrow$)
We prove it by contraposition. Suppose that $\rho$ fails a quantum-Schnorr test, $(S^{r})_r$ at $\epsilon$. For all $j$, let $s_j$ be the least $t$ such that \[\sum_{i=0}^{t} \tau (S^{i}) > \alpha -2^{-j}.\]
We show how the sequence $(s_j)_j$ can be computed. First, let $\sum_{r}\tau(S^{r})=\alpha$, be a computable real which is not a dyadic rational. $s_j$ may be computed as follows: Note that as $\alpha$ is not a dyadic rational but $\tau(S^i)$ is a dyadic rational for all $i$, we have that
\[\sum_{i=0}^{s_{j}-1} \tau (S^{i}) < \alpha-2^{-j} < \sum_{i=0}^{s_{j}} \tau (S^{i}).\]
By Proposition 5.1.1 in \cite{misc1}, the left cut, $L(\alpha-2^{-j})$ of $\alpha-2^{-j}$ is computable. So, we may search for rationals $q \in L(\alpha-2^{-j}), q' \notin L(\alpha-2^{-j}) $ and for a $t$ such that,
\[\sum_{i=0}^{t-1} \tau (S^{i}) \leq q<q'\leq \sum_{i=0}^{t } \tau (S^{i}).\]
This $t$ is the needed $s_j$.
Now, let $\alpha$ be a dyadic rational. Then, $\alpha - 2^j$ has a finite binary representation and $s_j$ can be directly computed. So, in summary, the $(s_j)_j$ is a computable sequence, after (non-uniformly) knowing whether $\alpha$ is a dyadic rational or not. 
For all $r\geq 0$, define special projections
\[ G_{r}:= \sum_{i=s_{r}+1}^{s_{r+1}}S^{i}.\]
So,
\[ \tau(G_{r}) \leq \sum_{i=s_{r}+1}^{\infty}\tau(S^{i}) = \alpha -\sum_{i=0 }^{s_{r}}\tau(S^{i}) < 2^{-r} .\]

Notation: Let each $S^i$ be an operator on $ \mathbb{C}^{2^{n_{i}}}$. Let $n_{r} =$ max$\{n_{i}:s_{r}+1 \leq i \leq s_{r+1}\}$. By tensoring with the identity, we may assume that all
$S^i$, for $s_{r}+1 \leq i \leq s_{r+1}$, are operators on $\mathbb{C}^{2^{n_{r}}}$. Let $F_r \subseteq \mathbb{C}^{2^{n_{r}}}$
be an orthonormal set of complex algebraic vectors spanning the range of $G_r$.
Define a computable measure machine, $C$ as follows. 
On input $0^{r}10$, $C$ outputs $F_{2r}$ and on input $0^{r}11$, $C$ outputs $F_{2r+1}$. $C$ is clearly prefix-free and the measure of its domain is $\sum_{r}2^{-r+2}$, which is computable. Since each $G_r$ is a finite sum of the $S^i$s and as $\rho$ fails $(S^i)_i$ at $\epsilon$, there exist infinitely many $r$ such that Tr$(\rho_{n_{r}}G_{r})>\epsilon$. Since  we may let $n_r$ be strictly increasing in $r$, there are infinitely many such $n_r$. Fix such an $n_r$ and let $x= \lfloor r/2 \rfloor$ (I.e., $r=2x$ or $r=2x+1$). Then, QK$^{\epsilon}_{C}(\rho_{n_{r}}) \leq x+2+n_{r}+\text{log}\tau(G_{r}) \leq x+2+n_{r}- 2x.$ So, $\lfloor r/2 \rfloor-2 \leq n_{r}- $QK$^{\epsilon}_{C}(\rho_{n_{r}}) $. Letting $r$ go to infinity completes the proof.

\end{proof}

Theorem \ref{thm:Chaitin} is a Chaitin-type characterization of quantum-Schnorr randomness using QK$^{\epsilon}_{C}$. Together with Theorem \ref{thm:677} and lemma 3.9 in \cite{bhojraj2020quantum}, it implies that Schnorr randoms have a Chaitin type characterization in terms of $K_C$ (Theorem \ref{class}). To the best of our knowledge, this is the first, albeit simple, instance where results in quantum algorithmic randomness are used to prove a new result in the classical theory.
\begin{thm}
\label{thm:Chaitin}
$\rho$ is quantum Schnorr random if and only if for all computable measure machines $C$ and all $\epsilon$, $\forall d \forall^{\infty} n$ QK$^{\epsilon}_{C}(\rho_{n})>n+d$.
\end{thm}
\begin{proof}
($\Rightarrow$)
Suppose toward a contradiction that there is a $C$, an $\epsilon>0$ and $c>0$ such that there are infinitely many $n$ with $ \text{QK}^{\epsilon}_{C}(\rho_{n})\leq n+c.$ Define a quantum Schnorr test $Q$ as follows. Let $T_{s}$ be the set of all $\sigma$ such that $C_{s}(\sigma)\downarrow = F_{\sigma}$, an orthonormal set such that $|\sigma|+$ log $|F_{\sigma}|< n_{\sigma} + c$ and $F_{\sigma} \subseteq  \mathbb{C}^{2^{n_{\sigma}}}$ for some $n_{\sigma}$. Let $T=\bigcup_{s}T_{s}.$ For all strings $\sigma$, let $P_{\sigma}:= \sum_{v\in F_{\sigma}} |v \big>\big<v |.$
Let $Q_s$ be the sequence of those $P_{\sigma}$ for $\sigma \in T_s$ and $Q$ the sequence of those $P_{\sigma}$ for $\sigma \in T$. That $Q$ is a quantum Schnorr test is shown by replacing $2^{-k}$ by $2^{-k-c}$ in the $\Longrightarrow$ direction of the proof of Theorem \ref{thm:schnor}. For any $n$ such that $ \text{QK}^{\epsilon}_{C}(\rho_{n})< n+c$, there is a $\sigma \in T$ such that Tr$(P_{\sigma} \rho_{n})>\epsilon$. So, $\rho$ fails $Q$ at $\epsilon$.

($\Leftarrow$) If $\rho$ is not quantum Schnorr random then by Theorem \ref{thm:schnor}, there is a $C$ and an $\epsilon$ such that $\forall d \exists n$ such that QK$^{\epsilon}_{C}(\rho_{n})\leq n-d$.
\end{proof}
We now show the classical version of Theorem \ref{thm:Chaitin}.
\begin{thm}
\label{class}
An infinite bitstring $X$ is quantum Schnorr random if and only if for all computable measure machines $C$, $\forall d \forall^{\infty} n$ K$_{C}(X\upharpoonright n)>n+d$.
\end{thm}
\begin{proof}
$(\Longrightarrow):$ Suppose first that $X$ is Schnorr random. Then, $\rho:=\rho_X$, the state induced by $X$ is quantum Schnorr random by lemma 3.9 in \cite{bhojraj2020quantum}. Suppose for a contradiction that there is a $C$ and a $d$ such that $\exists^{\infty} n$ such that $K_C(X\upharpoonright n) \leq n+d$. By Remark \ref{rem:qkc}, $\exists^{\infty} n$ such that $QK^{0.5}_C(\rho_n) \leq n+d$, contradicting Theorem \ref{thm:Chaitin}.
($\Longleftarrow): $ Suppose that $X$ is not Schnorr random. Once again, by lemma 3.9 in \cite{bhojraj2020quantum}, we have that $\rho:=\rho_X$ is not quantum Schnorr random. By Theorem \ref{thm:Chaitin}, there is a $C$, an $\epsilon$ and a $d$ such that $\exists^{\infty} n$ such that  QK$^{\epsilon}_{C}(\rho_{n}) \leq n+d$. By Theorem \ref{thm:677}, there is a $P$ such that $\exists^{\infty} n$ such that  $K_{P}(\rho_{n}) \leq n+d+1$, a contradiction.
\end{proof}
\section{Conclusion}
With the intent of developing a quantum version of $K$, we introduced $QK$, a notion of descriptive complexity for density matrices using classical prefix-free Turing machines. Many connections between $K$ and Solovay and Schnorr randomness in the classical theory turned out to have analogous connections connections between $QK$ and weak Solovay and quantum Schnorr randomness. 

To the best of our knowledge, the current paper is the only one to study the incompressibility of initial segments (in the sense of prefix-free classical Turing machines) of weak Solovay and quantum Schnorr random states. Nies and Scholz have explored connections between quantum Martin-L{\"o}f randomness and a version of $QC$ using unitary (quantum) machines\cite{unpublished}. 

An important open question is whether weak Solovay random states have a Levin\textendash Schnorr characterization in terms of $QK$. Techniques similar to those used in subsection \ref{subsect:weak} may prove to be useful in answering this.

It still remains to find a complexity based characterization of q-MLR states\cite{unpublished}. An important question is whether weak Solovay randomness is equivalent to q-MLR, a positive answer to which will yield a $QK$ based characterizations (namely, those in Theorems \ref{thm:7} and \ref{thm:39}) of q-MLR.

Abbott, Calude and Svozil use value indefiniteness to show that infinite bitstrings resulting from measuring a finite dimensional quantum system satisfy some notions of randomness \cite{DBLP:conf/birthday/AbbottCS15,DBLP:journals/mscs/AbbottCS14,qrng2020 ,abbott2014value}.

Using states instead of a finite dimensional system, and techniques entirely different from theirs, we showed that it is possible to generate a strong form of classical randomness from a quantum source which is not quantum random. More precisely, we construct a computable, non q-MLR state which yields an arithmetically random bitstring with probability one when `measured'\cite{qpl}. Arithmetic randomness is a strong form of classical randomness, strictly stronger than 
Martin-L{\"o}f randomness (See 6.8.4 in \cite{misc1}). Roughly speaking, states which yield a Martin-L{\"o}f random  bitstring with probability one when `measured using a computable basis' are defined to be \emph{measurement random} (See \cite{qpl} for the precise definition). It would be interesting to explore the relationship between the initial segment $QK$ complexity and measurement randomness of states.

\section{Acknowledgements}
I am indebted to my PhD thesis advisor, Joseph S. Miller for his encouragement, advice and support. I thank Andr{\'e} Nies for encouraging my study of quantum algorithmic randomness.
\bibliographystyle{plain}
\bibliography{references.bib}

\end{document}